\documentclass[11pt,a4paper]{article}
\usepackage[russian,greek,english]{babel}
\usepackage{polski}
\usepackage{todonotes}
\usepackage[utf8x]{inputenc}
\usepackage{amsmath,ifthen,color,eurosym}
\usepackage{wrapfig,hyperref,cite}

\usepackage{latexsym,amsthm,amsmath,amssymb,url}
\usepackage{graphicx}
\usepackage{enumitem}
\setlist{leftmargin=13.4mm}
\usepackage{framed}
 \usepackage{xifthen}

\newtheorem{lemma}{Lemma}
\newtheorem{theorem}{Theorem}

\newtheorem{observation}{Observation}
\newtheorem{lclaim}{Claim}[lemma]

\newcommand{\Oh}{\mathcal{O}}

\newcommand{\Arc}{{\sf arc}}

\newcommand{\defproblemu}[3]{
  \vspace{1mm}
\noindent\fbox{
  \begin{minipage}{0.95\textwidth}
  #1 \\
  {\bf{Input:}} #2  \\
  {\bf{Question:}} #3
  \end{minipage}
  }
  \vspace{1mm}
}

\newcommand{\defparproblem}[4]{
  \vspace{1mm}
\noindent\fbox{
  \begin{minipage}{0.96\textwidth}
  \begin{tabular*}{\textwidth}{@{\extracolsep{\fill}}lr} #1  & {\bf{Parameter:}} #3 \\ \end{tabular*}
  {\bf{Input:}} #2  \\
  {\bf{Question:}} #4
  \end{minipage}
  }
  \vspace{1mm}
}

\title{On the Parameterized Complexity of Graph Modification to first-order logic Properties\thanks{The two first authors have been supported by the Research Council of Norway via the projects ``CLASSIS'' and ``MULTIVAL". The third author has been supported  by  projects DEMOGRAPH (ANR-16-CE40-0028) and ESIGMA (ANR-17-CE23-0010). All  authors have been supported by the Research Council of Norway and the French Ministry of Europe and Foreign Affairs, via the Franco-Norwegian project PHC AURORA 2019.
\newline Emails of authors: \textsf{\{fedor.fomin, petr.golovach\}@ii.uib.no}, \textsf{sedthilk@thilikos.info}.}}

\author{
Fedor V. Fomin\thanks{
Department of Informatics, University of Bergen, Norway.} \addtocounter{footnote}{-1}
\and
Petr A. Golovach\footnotemark{}
\and 
Dimitrios M. Thilikos\thanks{AlGCo project team, CNRS, LIRMM, Universit\'e de Montpellier, Montpellier, France.}}

\begin{document}
\maketitle
\begin{abstract}
\noindent We establish  connections between parameterized/kernelization complexity of graph modification problems and expressibility in logic. For a first-order logic formula $\varphi$, we consider the problem of deciding whether an input graph can be modified by removing/adding at most $k$ vertices/edges such that the resulting modification has the property expressible by $\varphi$. We provide sufficient and necessary conditions on the structure of the prefix of $\varphi$ specifying when the corresponding graph modification problem is fixed-parameter tractable (parameterized by $k$) and when it admits a polynomial kernel.
 \end{abstract}


\section{Introduction}\label{writings}
A  variety of algorithmic graph problems, called {\em modification problems}, can be formulated as problems of modifying a graph such that the resulting graph  satisfies some fixed desired property. 
 The study of graph modification problems is one of the most popular trends in graph algorithms, and in particular, in parameterized complexity. 
 One of the classic results about graph modification problems  is the work of Lewis and Yannakakis \cite{lewis1980nodedeletion}, which  provides necessary and sufficient conditions   (assuming  ${\sf P} \neq{\sf NP}$) of  polynomial time solvability of vertex-removal problems for hereditary properties. For other types of graph modification problems, like edge-removal problems \cite{yannakakis1981edgedeletion}, no such dichotomy is known.  
 For the past 30 years  graph modification problems   served as a strong inspiration for developing
  new  methods and techniques   in 
 {parameterized/kernelization algorithms and complexity}, see the books \cite{CyganFKLMPPS15,DowneyF13,FlumG06,Niedermeierbook06} for an overview of the area.

 In this paper we approach   graph modification problems from the perspective of descriptive complexity. Descriptive complexity is the field of logic which studies the relations 
 between computational complexity  and  expressibility  in logic.  The classic example  of a theorem in descriptive complexity is the theorem of Fagin~\cite{fagin7generalized}
 asserting that a property of graphs is in {\sf NP} if and only if it is definable by an existential second-order formula.  
 We refer to the recent book of   
 Grohe  \cite{grohe2017descriptive} for a modern overview of descriptive complexity. 
 The significant amount of research in descriptive complexity is devoted to the study of prefix classes of certain logics. A prefix class is a syntactic fragment of first-order or second-order logic with formulas in prenex normal form and imposed constrains on the patterns of quantifiers in formulas.  For example, the  study of prefix classes of first-order logic is provided in the book of 
 B{\"o}rger,  Gr{\"a}del,   and Gurevich  
 \cite{borger2001classical},
see also the work of 
  Gottlob, Kolatis and Schwentick   \cite{DBLP:journals/jacm/GottlobKS04} on characterizing the computational complexity of  prefix classes of second-order logic.

    \medskip\noindent\textbf{Our results.}
  Let $\phi$ be  an FOL  formula on (undirected) graphs in  
 {prenex normal form}.  In particular,    $\phi={\tt Q}_{1}x_{1}{\tt Q}_{2}x_{2}\cdots{\tt Q}_{t}x_{t} \chi$, where $t$ is some constant,  each ${\tt Q}_i\in\{\forall,\exists\}$ is a quantifier,  $x_i$ is a variable, and $\chi$ is a quantifier-free part that depends on the variables $x_1,\ldots,x_t$.
We consider the following generic problems (we use ``$-$'' for the vertex/edge removal, ``$+$'' for the edge addition and ``$\bigtriangleup$'' for the symmetric difference). \smallskip
 
  \defparproblem{\sc Vertex-Removal to $\phi$}{A graph $G$ and an integer $k$.}{$k$}{ Does there exist  a vertex set $S\subseteq V(G)$ with $|S|\leq k$ such that 
$G- S\models \phi$?}

\defparproblem{\sc Edge-Removal to $\phi$}{A graph $G$ and an integer $k$.}{$k$}{ Does there exist  an edge set $F\subseteq E(G)$ with $|F|\leq k$ such that 
$G- F\models \phi$?
}

 \defparproblem{\sc Edge-Completion to $\phi$}{A graph $G$ and an integer $k$.}{$k$}{ Does there exist $F\subseteq {V(G)\choose 2}\setminus E(G)$ with $|F|\leq k$ such that 
$G+ F\models \phi$?}
 
  \defparproblem{\sc Edge-Editing to $\phi$}{A graph $G$ and an integer $k$.}{$k$}{ Does there exist $F\subseteq {V(G)\choose 2}$ with $|F|\leq k$ such that 
$G\bigtriangleup F\models \phi$?}

 \medskip\medskip 
\noindent  For example, for $\phi=\forall u\forall v  \neg (u \sim v)$,  
{\sc Vertex-Removal to $\phi$} is equivalent to {\sc Vertex Cover} that is the graph modification problem asking whether one can  remove at most $k$ vertices such that the resulting graph has no edges (we use $u\sim v$ for the adjacency predicate).
More generally, any vertex-removal problem to a  graph class characterized by a finite set of forbidden subgraphs, can be expressed  as {\sc Vertex-Removal to $\phi$} for some $\phi$ with only  $\forall$ quantifications over variables, where the number of variables is the maximum number of vertices of a forbidden graph. Clearly, using FOL, we are able to express other properties. For example, the property that the diameter of a graph is at most two cannot be expressed using forbidden subgraphs but can easily be written as the FOL formula $\forall u\forall v\exists w [(u=v)\vee (u\sim v)\vee ((u\sim w)\wedge(v\sim w))]$.
Similarly, the edge variants of modification problems to $\phi$ capture quite a few interesting and well-studied problems like 
  {\sc Cluster Editing}, where the task is to change at most $k$ adjacencies in the graph  resulting in a disjoint union of cliques.

We consider 
 modification problems,  where the specification of a  prefix class of  formula $\phi$ is defined  
according
 to  the  arithmetic  hierarchy (also known as  Kleene-Mostowski  hierarchy) used for classifications of the formulas in the  first-order arithmetic language (see, e.g., \cite{Smorynski77}).  
We define prefix classes    according to alternations of quantifiers, that is, switchings from $\forall$ to $\exists$ or vice versa in the prefix string of the formula. We allow a formula to have 
 \emph{free}, i.e.,  non-quantified, variables.
Let $\Sigma_0=\Pi_0$  be the classes of FOL-formulas without quantifiers. 
For a positive integer $s$, the class $\Sigma_s$ contains formulas that could be written in the form 
\[\phi=\exists x_1\exists x_2\cdots \exists x_t \psi,\]
where $\psi$ is  a $\Pi_{s-1}$-formula,  $t$ is some integer,  and $x_1,\ldots,x_t$ are free variables of $\psi$. Respectively, $\Pi_s$ consists of formulas 
\[\phi=\forall x_1\forall x_2\cdots \forall x_t \psi,\]
where $\psi$ is  a $\Sigma_{s-1}$-formula and $x_1,\ldots,x_t$ are free variables of $\psi$.
Note that we  allow $t=0$, which implies that for $s'>s$, $\Sigma_{s}\cup\Pi_s\subseteq \Sigma_{s'}\cap\Pi_{s'}$.

 We establish  a number of  algorithmic results about modification problems where the target property is  definable in  FOL. We complement these results by lower bounds, which in combination provide  a  neat 
   dichotomy  theorems about the parameterized complexity of such problems. Hence we establish sufficient and necessary conditions on the prefix classes of  FOL-formulas such that the corresponding  
    graph modification problems  are  fixed-parameter tractable and/or  admit a polynomial kernel.\medskip

Our first result shows the following {dichotomy} (subject to ${\sf W}[2]\neq{\sf FPT}$) for  {\sc Vertex-Removal to $\phi$}, depending on the structure of the prefix class of  $\phi$.

\begin{theorem}\label{magnetic}
~
\begin{itemize}
\item[(i)] For every $\phi\in \Sigma_3$ without free variables,  {\sc Vertex-Removal to $\phi$} is {\sf FPT}.
\item[(ii)] There is $\phi\in\Pi_3$ without free variables such that  {\sc Vertex-Removal to $\phi$} is ${\sf W}[2]$-hard.
\end{itemize}
\end{theorem}

In other words,   if the prefix of an FOL-formula $\phi$ has at most two alternations of quantifiers and, in the case of exactly two alternations, if the first quantifier is $\exists$, then 
{\sc Vertex-Removal to $\phi$} is {\sf FPT}. For each other type of quantifier alternations,   there exists a formula for which the problem becomes ${\sf W}[2]$-hard.\medskip

 For kernelization complexity of 
 {\sc Vertex-Removal to $\phi$}, we establish the following dichotomy. 

 \begin{theorem}\label{consiste}
~
\begin{itemize}
\item[(i)] For every $\phi\in \Sigma_1\cup\Pi_{1}$ without free variables,  {\sc Vertex-Removal to $\phi$} admits a polynomial kernel.
\item[(ii)] There is $\phi\in\Sigma_2$ ($\phi\in\Pi_2$) without free variables such that  {\sc Vertex-Removal to $\phi$} admits no polynomial kernel unless  ${\sf NP}\subseteq{\sf coNP}/{\sf poly}$.
\end{itemize}
\end{theorem}

\medskip
For edge-modification problems we prove the following. 
 
 \begin{theorem}\label{narcotic} ~
 \begin{itemize}
 \item[(i)] For every  $\phi\in \Sigma_2$ without free variables,   {\sc Edge-Removal to $\phi$},   {\sc Edge-Completion to $\phi$}, and {\sc Edge-Editing to $\phi$} are  {\sf FPT}.
 \item[(ii)] There exists $\phi\in \Pi_2$ without free varaibles  such that  {\sc Edge-Removal to $\phi$} (respectively,  {\sc Edge-Completion to $\phi$} and {\sc Edge-Editing to $\phi$}) is    ${\sf W}[2]$-hard.
 \end{itemize} 
 \end{theorem}
  
  We observe that if $\phi\in\Sigma_1$, then all   considered problems can be solved in polynomial time. Clearly, this means that they are {\sf FPT} and have trivial polynomial kernels. We complement this with lower bounds and summarize these results in the following theorem.

  \begin{theorem}\label{unguided} ~
\begin{itemize}
 \item[(i)] For every $\phi\in \Sigma_1$ without free variables,   {\sc Edge-Removal to $\phi$},   {\sc Edge-Completion to $\phi$}, and {\sc Edge-Editing to $\phi$} admit    polynomial kernels.
 \item[(ii)] There exists $\phi\in \Pi_1$ without free variables such that  {\sc Edge-Removal to $\phi$} (respectively,   {\sc Edge-Completion to $\phi$} and {\sc Edge-Editing to $\phi$}) admits  no polynomial kernel unless  ${\sf NP}\subseteq{\sf coNP}/{\sf poly}$.
 \end{itemize}
 \end{theorem}

This paper is organized as follows. In Section~\ref{consists}, we introduce basic notions and state some auxiliary results. In Section~\ref{skeleton}, we obtain the algorithmic upper bounds, that is, we show the claims (i) of Theorems~\ref{magnetic}--\ref{unguided}. In Section~\ref{immature}, we complement these results by the lower bounds given in the claims (ii) 
of Theorems~\ref{magnetic}--\ref{unguided}. We conclude with Section~\ref{barracks}, where we discuss some possible extension of our results and mention some directions for further research. 

 \section{Preliminaries}\label{consists}

\paragraph{Sets.} 
We use $\Bbb{N}$ to denote the set of all non-negative numbers.
Given some $k\in\Bbb{N}$, we denote $[k]=[1,k]$. 
Given a set $A$, we denote by $2^{A}$ the set of all its subsets and we define 
${A \choose 2}:=\{e\mid e\in 2^{A}\wedge |e|=2\}$.
We denote by $\mathbf{a}=\langle a_1,\ldots,a_r\rangle$ a sequence of elements of a set $A$ and call $\mathbf{a}$ an \emph{$r$-tuple} of simply a \emph{tuple}. 
Note that the elements of $\mathbf{a}$ not necessarily pairwise distinct. We denote by $\mathbf{a}\mathbf{b}$ the \emph{concatenation} of tuples $\mathbf{a}$ and $\mathbf{b}$.

 \paragraph{Graphs.}
All graphs in this paper are undirected, loop-less, and without multiple edges unless it is explicitly specified to be different.  
Given a graph $G$, we denote by $V(G)$ its vertex set and by $E(G)$ its edge set. For an edge $e=\{x,y\}\in E(G)$, we use instead the notation $e=xy$, that is equivalent to  $e=yx$.
We  denote $|G|=|V(G)|$. Throughout the paper we use $n$ to denote $|G|$ if it does not create confusion.
For a vertex $v$, $d_G(v)$ denotes the degree of $v$. 
For any set of vertices $S\subseteq V(G)$, we denote by $G[S]$ the subgraph of $G$ induced by the vertices from $S$. We also define $G-S:=G[V(G)\setminus S]$. Given an edge set $F\subseteq E(G)$, we denote $G- F=(V(G),E(G)\setminus F)$. Also, given a set  $F\subseteq {V(G\choose 2}\setminus E(G)$, i.e., $F$ is a set of pairs of vertices that are not edges of $G$,  we define $G+F=(V(G),E(G)\cup F)$, and for $F\subseteq {V(G\choose 2}$,  we define $G\bigtriangleup F=(V(G),E(G)-E(G)\cap F+F\setminus E(G))$.

\paragraph{Formulas.}
In this paper we deal with logic formulas on graphs. In particular we will deal with formulas of first-order logic (FOL).  The syntax of FOL-formulas on graphs includes the logical connectives 
$\vee$, $\wedge$, $\neg$, variables for vertices, 
the quantifiers $\forall$, $\exists$ that are applied to these variables, the predicate $u \sim v$, where $u$ and $v$ are vertex variables and the interpretation is that $u$ and $v$ are adjacent, and the  equality of variables representing vertices. .
It also convenient to assume that we have the logical connectives $\rightarrow$ and $\leftrightarrow$.
An FOL-formula $\phi$ is  in {\em prenex normal form} if it is written as $\phi={\tt Q}_{1}x_{1}{\tt Q}_{2}x_{2}\cdots{\tt Q}_{t}x_{t} \chi$  where 
each ${\tt Q}_i\in\{\forall,\exists\}$ is a quantifier,  $x_i$ is a varible, and $\chi$ is a quantifier-free part that depends on the variables $x_1,\ldots,x_t$. Then
${\tt Q}_{1}x_{1}{\tt Q}_{2}x_{2}\cdots{\tt Q}_{t}x_{t}$ is
 referred as the \emph{prefix} of $\phi$. 
From now on, when we mention the term ``FOL-formula'', we mean an FOL-formula on graphs that is in prenex normal form. 
For an FOL-formula $\phi$  without free variables and a graph $G$, we write $G\models \phi$ to denote that $\phi$ evaluates to \emph{true} on $G$. 

For technical reasons, we extend FOL-formulas on graphs to  structures of a special type. We say that a pair $(G,\mathbf{v})$, where $\mathbf{v}=\langle v_1,\ldots,v_r\rangle$ is an $r$-tuple of vertices of $G$, is an \emph{$r$-structure}. Let $\phi$ be an FOL-formula without free variables and let $\mathbf{x}=\langle x_1,\ldots,x_r\rangle$ be an $r$-tuple of pairwise distinct variables  of $\phi$. We denote by $\phi[\mathbf{x}]$ the formula obtained from $\phi$ by the deletion of the quantification over $x_1,\ldots,x_r$, that is, these variables become the free variables of $\phi[\mathbf{x}]$. For an $r$-structure $(G,\mathbf{v})$ with $\mathbf{v}=\langle v_1,\ldots,v_r\rangle$ and $\phi[\mathbf{x}]$, we write $(G,\mathbf{v})\models \phi[\mathbf{x}]$ to denote that $\phi[\mathbf{x}]$ evaluates to \emph{true} on $G$ if $x_i$ is assigned $v_i$ for $i\in[r]$. If $r=0$, that is, $\mathbf{v}$ and $\mathbf{x}$ are empty, then $(G,\mathbf{v})\models \phi[\mathbf{x}]$ is equivalent to $G\models \phi$.

\paragraph{Parameterized Complexity.}  We refer to the  books~\cite{CyganFKLMPPS15,DowneyF13,FlumG06,Niedermeierbook06} for the detailed introduction to the field. Here we only 
briefly review the basic notions.

Parameterized Complexity is a  bivariate framework
for studying the computational complexity of computational problems. One variable is the input size~$n$ and the other is a \emph{parameter}~$k$ associated with the input. 
The main goal is to confine the combinatorial explosion in the running time of an algorithm for an {\sf NP}-hard problem to depend only on $k$.
Thus, a parameterized problem is defined formally as a language $\mathcal{L}\subseteq\Sigma^*\times\mathbb{N}$, where $\Sigma^*$ is a set of string over a finite alphabet $\Sigma$.
A parameterized problem is said to be \emph{fixed parameter tractable} (or {\sf FPT}) if it can be solved in time $f(k)\cdot n^{\Oh(1)}$ for some computable function~$f$. Also, we say that a parameterized problem belongs in the class ${\sf XP}$  if it can be solved in time $O(n^{f(k)})$ for some computable function~$f$. 
The complexity class {\sf FPT} consists of all fixed parameter tractable problems.
Parameterized complexity theory also provides tools to disprove the existence of {\sf FPT} algorithms under plausible complexity-theoretic assumptions. For this, Downey and Fellows introduced a hierarchy of parameterized complexity classes, namely  ${\sf FPT}\subseteq {\sf W}[1]\subseteq  {\sf W[2]}\subseteq  {\sf W[3]}\subseteq \cdots\subseteq {\sf XP}$ and conjectured that it is proper. This conjecture plays a central role in obtaining lower complexity bounds. The basic way to show that it is unlikely that a parameterized problem admit an {\sf FPT} algorithm is to show that it is ${\sf W}[1]$ or ${\sf W}[2]$-hard using a \emph{parameterized reduction} form a known ${\sf W}[1]$ or ${\sf W}[2]$-hard problem. \\

A \emph{kernelization} for a parameterized problem is a polynomial time algorithm that maps each instance $(I,k)$ of a parameterized problem with the input~$I$ and parameter~$k$ to an instance $(I',k')$ of the same problem such that
\begin{itemize}
\item[(i)] $(I,k)$ is a
{\sf yes}-instance if and only if $(I',k')$ is a {\sf yes}-instance, and
\item[(ii)] $|I'|+k'$ is bounded by~$f(k)$ for some computable function~$f$.
\end{itemize}
The output $(I',k')$ is called a \emph{kernel}. The function~$f$ is said to be the \emph{size} of the kernel. A kernel is \emph{polynomial} if~$f$ is polynomial.
While it can be shown that every decidable parameterized problem is {\sf FPT} if and only if it admits a kernel, it is unlikely that every problem in {\sf FPT} has a polynomial kernel.
In particular, the now standard \emph{composition} and \emph{cross-composition} techniques~\cite{BodlaenderDFH09,BodlaenderJK14} 
allow to show that certain problems have no polynomial kernels  unless  ${\sf NP}\subseteq{\sf coNP}/{\sf poly}$.

\medskip
To solve all considered problems, we have to solve  the {\sc Model Checking} problem for first-order logic on graphs:\medskip

  \defproblemu{\textsc{Model Checking}}{A graph $G$ and an FOL-formula $\phi$.}{Does $G\models \phi$?}\medskip

{\sc Model Checking} is known to be ${\sf PSPACE}$-complete~\cite{Vardi82}. The problem is also hard from the parameterized complexity viewpoint when parameterized by the size of the formula. It was proved by Frick and Grohe in~\cite{FrickG04} that the problem is ${\sf AW}[*]$-complete for this parametrization (see, e.g., the book~\cite{FlumG06} for the definition of the class). Thus, it is unlikely that {\sc Model Checking} is {\sf FPT} when parameterized by the formula size. This immediately implies that the problem {\sc Vertex-Removal to $\phi$} as well as  
the problems {\sc Edge-Removal/Completion}/{\sc Editing to $\phi$} are ${\sf AW}[*]$-hard when parameterized by the size of $\phi$ even for $k=0$. However, {\sc Model Checking} is in {\sf XP} when parameterized by the number of variables. In particular, if $\phi$ has $r$ variables and input size is $n$, then it can be solved in time $\Oh(n^r)$ by exhaustive search. The currently best algorithm is given by Williams in~\cite{Williams14} who proved the following.

\begin{theorem}[\!\!\cite{Williams14}]\label{appetite}
{\sc Model Checking} can be solved in time $\tilde{\Oh}(n^\omega)$ for formulas with 3 variables and if the number of variables $r\geq 3$, then it can be solved in time  $\tilde{\Oh}(n^{r-3+\omega})$ where $\omega$ is the matrix-multiplication exponent. Moreover, if $r\geq 9$, then {\sc Model Checking} can be solved in time $n^{r-1+o(1)}$.
\end{theorem}

Here $\tilde{\Oh}(f(n))$ is used to denote an upper bound $\Oh(f(n)\log^c n)$ for some positive constant $c$. These algorithms are, in fact, asymptotically optimal up to the \emph{Strong Exponential Time Hypothesis} (SETH) (see~\cite{ImpagliazzoPZ01,CyganFKLMPPS15} for the definition). It was shown by Williams~\cite{Williams14} that if {\sc Model Checking} for formulas with $r\geq 4$ can be solved in time $\Oh(n^{r-1-\varepsilon})$ for some $\varepsilon>0$, then SETH is false.

Because of these results, we assume throughout the paper that the FOL-formulas in the considered modification problems have a \emph{constant} number of variables and, therefore, constant sizes. In particular, the exponents of polynomials in  running times and the sizes of kernels should depend on the length $|\phi|$ of the formula $\phi$.

\smallskip
We conclude this section by observing that Theorems~\ref{narcotic} and \ref{unguided} claim the same complexity status for  \textsc{Edge-Removal to $\phi$} and \textsc{Edge-Completion to $\phi$}.
 This is not surprising, because these problems are equivalent in the following sense. Denote by $\overline{G}$ the \emph{complement} of a graph $G$, that is, the graph with the same vertex set such that every two distinct vertices are adjacent in $\overline{G}$ if and only if they are nonadjacent in $G$. For an FOL-formula $\phi$, denote by $\overline{\phi}$ the formula obtained from $\phi$ by replacing each adjacency predicate by the subformula expressing non-adjacency of distinct vertices, that is, $u\sim v$ is replaced by $\neg(u=v)\wedge\neg(u\sim v)$. Then we can make the following straightforward observation.

\begin{observation}\label{adorning}
For every FOL-formula $\phi$, $(G,k)$ is a {\sf yes}-instance of \textsc{Edge-Removal to $\phi$} if and only if $(\overline{G},k)$ is a a {\sf yes}-instance of
\textsc{Edge-Completion to $\overline{\phi}$}.
\end{observation} 

By Observation~\ref{adorning}, it is sufficient to show  Theorems~\ref{narcotic} and \ref{unguided} for  \textsc{Edge-Removal to $\phi$} and  \textsc{Edge-Editing to $\phi$}.

\section{Upper bounds}\label{skeleton}
In this section we prove the claims (i) of Theorems~\ref{magnetic}--\ref{unguided}. 
First, we consider \textsc{Vertex-Removal to $\phi$}.

\begin{lemma}\label{friesian}
For every $\phi\in \Sigma_3$ without free variables, \textsc{Vertex-Removal to $\phi$} can be solved in time $|\phi|^k\cdot n^{\Oh(|\phi|)}$. 
\end{lemma}

\begin{proof}
Consider an instance $(G,k)$ of \textsc{Vertex-Removal to $\phi$} for 
$$\phi=\exists x_1\cdots \exists x_r\forall y_1\cdots \forall y_s\exists z_1\cdots\exists z_t\chi,$$
  where $r,s,t\geq 0$ and $\chi$ is quantifier-free. Let $\mathbf{x}=\langle x_1,\ldots,x_r\rangle$, $\mathbf{y}=\langle y_1,\ldots,y_s\rangle$,  and 
$\mathbf{z}=\langle z_1,\ldots,z_t\rangle$.

Assume that $(G,k)$ is a {\sf yes}-instance of \textsc{Vertex-Removal to $\phi$}. This means that there is $S\subseteq V(G)$ of size at most $k$ such that $G-S\models \phi$. 
Observe that $G-S\models \phi$ if and only if there is an $r$-tuple $\mathbf{u}=\langle u_1,\ldots,u_r\rangle$ of vertices of $G-S$ such that  $(G-S,\mathbf{u})\models\phi[\mathbf{x}]$. 

We use this observation, and for each $r$-tuple $\mathbf{u}$ of vertices of $G$, check whether there is $S\subseteq V(G)$ of size at most $k$ that has no common vertices with $\mathbf{u}$ and it holds that $(G-S,\mathbf{u})\models\phi[\mathbf{x}]$. If we find such a set $S$, we return this solution for the considered instance of \textsc{Vertex-Removal to $\phi$}. Otherwise, if we fail to find $S$ for all $r$-tuples $\mathbf{u}$, we conclude that $(G,k)$ is a {\sf no}-instance. 
From now we assume that $\mathbf{u}$ is given. 

Suppose that $(G,\mathbf{u})\models \phi[\mathbf{x}]$ does not hold. Then there is an $s$-tuple $\mathbf{v}=\langle v_1,\ldots,v_s\rangle$ of vertices of $G$ such that  $(G,\mathbf{u}\mathbf{v})\models \phi[\mathbf{x}\mathbf{y}]$ does not hold.  Our algorithm is based on the following crucial claim.

\begin{lclaim}\label{granting}
For every $S\subseteq V(G)$ such that $S$ is disjoint with $\mathbf{u}$ and $(G-S,\mathbf{u})\models \phi[\mathbf{x}]$, 
 $S$  contains at least one vertex of $\mathbf{v}$.
\end{lclaim}

The proof is by contradiction. Assume that $(G-S,\mathbf{u})\models \phi[\mathbf{x}]$ but $S$ and $\mathbf{v}$ are disjoint. Then $(G-S,\mathbf{uv})\models \phi[\mathbf{xy}]$. By definition, this means that there is a $t$-tuple of vertices $\mathbf{w}$ of $G-S$ such that $(G-S,\mathbf{uvw})\models \phi[\mathbf{xyz}]$. In other words, $\chi$ evaluates to \emph{true}  if the $x$, $y$ and $z$-variables are assigned to $\mathbf{u}$, $\mathbf{v}$ and $\mathbf{w}$ respectively. This immediately implies that $(G,\mathbf{uvw})\models \phi[\mathbf{xyz}]$ and, therefore, 
$(G,\mathbf{uv})\models \phi[\mathbf{xy}]$. This contradicts the assumption that $(G,\mathbf{u}\mathbf{v})\not\models \phi[\mathbf{x}\mathbf{y}]$.
 
\medskip
Claim~\ref{granting} leads to the following  recursive  algorithm that find a solution $S$ for the given $\mathbf{u}$ (if such a solution exist). The algorithm receives as the input the current set $S$ that is initially set to be empty and finds a solution $S^*\supseteq S$ as follows.  

\begin{enumerate}
\item If $(G-S,\mathbf{uv})\models \phi[\mathbf{xy}]$ for all $s$-tuples $\mathbf{v}=\langle v_1,\ldots,v_s\rangle$ of vertices of $G-S$, then return $S$ and stop.
\item Otherwise, for an $s$-tuple $\mathbf{v}=\langle v_1,\ldots,v_s\rangle$ of vertices of $G-S$ such that $(G-S,\mathbf{u}\mathbf{v})\not\models \phi[\mathbf{x}\mathbf{y}]$, do the following:
\begin{itemize}
\item[(i)] if $|S|=k$ or all the vertices of $\mathbf{v}$ are in $\mathbf{u}$, then stop;
\item[(ii)] else, for each $v_i\in\mathbf{v}$ that is not in $\mathbf{u}$, call the algorithm for $S'=S\cup\{v_i\}$.
\end{itemize}
\end{enumerate}

The correctness of the algorithm follows from Claim~\ref{granting}. Concerning the running time of the algorithm.  At each iteration we check at most $n^s$ $s$-tuples $\mathbf{v}$ and for each $\mathbf{v}$ we verify in time $n^{\Oh(|\phi|)}$ whether  $(G-S,\mathbf{uv})\models \phi[\mathbf{xv}]$. Hence, each iteration takes  time $n^{\Oh(|\phi|)}$. Also at each iteration we branch into at  most $s$ subproblems and the depth of the search tree produced by the algorithm is at most $k$. Thus the running time is $s^k\cdot n^{\Oh(|\phi|)}$. Recall that we call the algorithm for each $r$-tuple $u$. Since there are $n^r$ such tuples, we have that the total running time is $s^k\cdot n^{\Oh(|\phi|)}$, which can be rewritten as $|\phi|^k\cdot n^{\Oh(|\phi|)}$.
\end{proof}

 Lemma~\ref{friesian} immediately implies Theorem~\ref{magnetic} (i).

We move to  {\sc Edge-Removal to $\phi$} and {\sc Edge-Editing to $\phi$}.

\begin{lemma}\label{segments}
For every $\phi\in \Sigma_2$ without free variables, \textsc{Edge-Removal to $\phi$} and \textsc{Edge-Editing to $\phi$} can be solved in time $|\phi|^{2k}\cdot n^{\Oh(|\phi|)}$. 
\end{lemma}

\begin{proof}
The proof is similar to the proof of Lemma~\ref{friesian}. We show the claim for \textsc{Edge-Removal to $\phi$}  and then explain how it should be modified for \textsc{Edge-editing to $\phi$}.

Let $(G,k)$ be an instance of  \textsc{Edge-Removal to $\phi$} for 
$$\phi=\exists x_1\cdots \exists x_r\forall y_1\cdots \forall y_s\chi,$$
where $\chi$ is quantifier-free. Let $\mathbf{x}=\langle x_1,\ldots,x_r\rangle$ and $\mathbf{y}=\langle y_1,\ldots,y_s\rangle$.

We observe that $F\subseteq E(G)$ of size at most $k$ is a solution for an instance $(G,k)$ of \textsc{Edge-Removal to $\phi$} if and only if  there is an $r$-tuple $\mathbf{u}=\langle u_1,\ldots,u_r\rangle$ of vertices of $G$ such that  $(G-F,\mathbf{u})\models\phi[\mathbf{x}]$. Respectively, for each $r$-tuple $\mathbf{u}$ of vertices of $G$, we check whether there is $F\subseteq E(G)$ of size at most $k$ such that  $(G-F,\mathbf{u})\models\phi[\mathbf{x}]$. If we find such $F$, we return this solution and we obtain that $(G,k)$ is a {\sf no}-instance otherwise. Assume that $\mathbf{u}$ is given.

If the property $(G,\mathbf{u})\models \phi[\mathbf{x}]$ is not fulfilled, then there is an $s$-tuple $\mathbf{v}=\langle v_1,\ldots,v_s\rangle$ of vertices of $G$ such that it  does not hold  that $(G,\mathbf{u}\mathbf{v})\models \phi[\mathbf{x}\mathbf{y}]$.  We use the following claim.

\begin{lclaim}\label{pensions}
For every $F\subseteq E(G)$ such that  $(G-F,\mathbf{u})\models \phi[\mathbf{x}]$,  
 $F$  contains at least one edge with both end-vertices in $\mathbf{uv}$.
\end{lclaim}

To obtain a contradiction, assume that $(G-F,\mathbf{u})\models \phi[\mathbf{x}]$ but every edge of $F$ has at least one end-vertex outside the tuples $\mathbf{u}$ and $\mathbf{v}$. This means that $\chi$ evaluates to \emph{true} on $G-F$ if the $x$ any $y$-variables are assigned to $\mathbf{u}$ and $\mathbf{v}$ respectively.  Notice that every two vertices of $\mathbf{uv}$ are adjacent in $G-F$ if and only if they are adjacent in $G$. Hence, $\chi$ evaluates to \emph{true} on $G$ if the $x$ any $y$-variables are assigned to $\mathbf{u}$ and $\mathbf{v}$ respectively. This means that $(G,\mathbf{u}\mathbf{v})\models \phi[\mathbf{x}\mathbf{y}]$; a contradiction.

\medskip
We construct the following  recursive branching algorithm that finds a solution $F$ for the given $\mathbf{u}$ if it exists. The algorithm takes as the input the current set $F$ that is initially empty and finds a solution $F^*\supseteq F$:  

\begin{enumerate}
\item If $(G-F,\mathbf{uv})\models \phi[\mathbf{xy}]$ for all $s$-tuples $\mathbf{v}=\langle v_1,\ldots,v_s\rangle$ of vertices of $G$, then return $F$ and stop.
\item Otherwise, for an $s$-tuple $\mathbf{v}=\langle v_1,\ldots,v_s\rangle$ of vertices of $G$ such that $(G-F,\mathbf{u}\mathbf{v})\not\models \phi[\mathbf{x}\mathbf{y}]$, do the following:
\begin{itemize}
\item[(i)] set $L\subseteq E(G)$ be the set of edges with both end-vertices  in $\mathbf{uv}$,
\item[(ii)] if $|F|=k$ or $L=\emptyset$, then stop;
\item[(iii)] else,  and  for each $e\in L$, call the algorithm for $F'=F\cup\{e\}$.
\end{itemize}
\end{enumerate}

The correctness of the algorithm follows from Claim~\ref{pensions}. 
   On each iteration we check at most $n^s$ $s$-tuples $\mathbf{v}$, and for each $\mathbf{v}$, verify in time $n^{\Oh(|\phi|)}$ whether  $(G-F,\mathbf{uv})\models \phi[\mathbf{xv}]$. Hence, each iteration can be done in time $n^{\Oh(|\phi|)}$. Also on each iteration we have at most $|L|\leq \binom{r+s}{2}$ branches and the depth of the search tree produced by the algorithm is at most $k$. This implies that the running time is $(r+s)^{2k}\cdot n^{\Oh(|\phi|)}$. Recall that we call the algorithm for each $r$-tuple $u$. Since there are $n^r$ such tuples, we have that the total running time is  $|\phi|^{2k}\cdot n^{\Oh(|\phi|)}$.

For \textsc{Edge-Editing to $\phi$}, the algorithm is essentially the same.
The difference is that in addition to edge removal we allowed to add edges.  
Respectively, we replace $G-F$ by $G\bigtriangleup F$ in the above algorithm and modify Steps 2 (i)---(iii):
\begin{itemize}
\item[(i)] set $L$ be the set of pairs of distinct  vertices of $\mathbf{uv}$,
\item[(ii)] if $|F|=k$ or $L=\emptyset$, then stop;
\item[(iii)] else,  and  for each $e\in L$, call the algorithm for $F'=F\cup\{e\}$.
\end{itemize}
Notice that the variant of Claim~\ref{pensions} , where  $G-F$ is replaces by $G\bigtriangleup F$, holds and this implies correctness.
The time analysis is the same.
\end{proof}

Lemma~\ref{segments} together with Observation~\ref{adorning} implies Theorem~\ref{narcotic}(i). Our next aim is show kernelization upper bounds. First, we observe that for $\Sigma_1$-formulas, our problems can be solved in polynomial time.

\begin{lemma}\label{presents}
For every $\phi\in \Sigma_1$ without free variables, \textsc{Vertex-Removal to $\phi$},    {\sc Edge-Removal to $\phi$},    and {\sc Edge-Editing to $\phi$} can be solved in time $n^{\Oh(|\phi|)}$. 
\end{lemma}

\begin{proof}
Assume that 
$$\phi=\exists x_1\cdots\exists x_r\chi$$
where $\chi$ is quantifier-free.

For \textsc{Vertex-Removal to $\phi$}, it is sufficient to observe that $(G,k)$ is a {\sf yes}-instance of the problem if and only if $G\models\phi$. We can use Theorem~\ref{appetite} and solve the problem in time $n^{\Oh{|\phi|}}$. 

For {\sc Edge-Removal to $\phi$} and {\sc Edge-Editing to $\phi$}, we can observe that $(G,k)$ is a {\sf yes}-instance if and only if there is a set of vertices $U$ of size $s=\min\{r,n\}$ 
such that $(G[U],k)$ is a {\sf yes}-instance. We can check all such sets $U$ in time $n^{\Oh(|\phi|)}$, and for each set, we use  brute force to verify whether  $(G[U],k)$ is a {\sf yes}-instance. Since the brute force checking of all subsets of edges or pairs of vertices of $G[U]$ of size at most $k'=\min\{k,\binom{s}{2}\}$ can be done in time $s^{2k'}$, the total running time is $s^{2s^2}\cdot n^{\Oh(|\phi|)}$. Because $s\leq|\phi|$, we can write it as $n^{\Oh(|\phi|)}$.
\end{proof}

Because every problem that can be solved in polynomial time has a trivial polynomial kernel, Lemma~\ref{presents} together with Observation~\ref{adorning} implies Theorem~\ref{unguided}(i). Clearly, the lemma also implies the claim of Theorem~\ref{consiste} (i) for $\Sigma_1$-formulas. It remains to prove it for $\Pi_1$-formulas. For this, we need the classic result of Lewis and Yannakakis~\cite{lewis1980nodedeletion}. A graph property $P$ is said to be \emph{hereditary} if for each graph $G$ satisfying $P$, it hold that $P$ holds for every induced subgraph of $G$. A property $P$ is \emph{nontrivial} if it is true for infinitely many graphs and it is false for infinitely many graphs. \textsc{Vertex-Removal to $P$} asks, given a graph $G$ and a positive integer $k$, whether it is possible to remove at most $k$ vertices of $G$ to obtain a graph satisfying $P$.
It was proved by Lewis and Yannakakis~\cite{lewis1980nodedeletion} that the following dichotomy holds for a hereditary property $P$ that can be tested in polynomial time: \textsc{Vertex-Removal to $P$} can be solved in polynomial time if $P$ is trivial, and the problem is {\sf NP}-complete otherwise.

\begin{lemma}\label{teaching}
For every $\phi\in \Pi_1$ without free variables,  {\sc Vertex-Removal to $\phi$} admits a polynomial kernel.
\end{lemma}

\begin{proof}
Let $(G,k)$ be an instance of {\sc Vertex-Removal to $\phi$}  for 
$$\phi=\forall x_1\cdots\forall x_r\chi$$
where $\chi$ is quantifier-free. Let $\mathbf{x}=\langle x_1,\ldots,x_r\rangle$.
Observe that the graph property $G\models\phi$ is hereditary for $\Pi_1$-formulas. If this property is trivial, we can solve {\sc Vertex-Removal to $\phi$} in polynomial time~\cite{lewis1980nodedeletion} and conclude that the problem admits a trivial polynomial kernel. Assume from now that the property $G\models\phi$ is not trivial. By the result of Lewis and Yannakakis~\cite{lewis1980nodedeletion}, {\sc Vertex-Removal to $\phi$} is {\sf NP}-complete.

For every tuple $\mathbf{v}$ of vertices of $G$, denote by $U_{\mathbf{v}}$ the set of vertices contained in $\mathbf{v}$. Let
$$\mathcal{U}=\{U_{\mathbf{v}}\mid \mathbf{v}\text{ is an }r\text{-tuple of vertices of  }G\text{ such that }(G,\mathbf{v})\not\models \phi[\mathbf{x}]\}.$$
The crucial observation is that $S\subseteq V(G)$ of size at most $k$ is a solution for $(G,k)$ if and only if $S$ is a \emph{hitting set} for $\mathcal{U}$, that is, $S\cap U_{\mathbf{v}}\neq\emptyset$ for every $U_{\mathbf{v}}\in\mathcal{U}$.  This observation is proved by the same arguments as Claim~\ref{granting} assuming that $\mathbf{u}$ is empty.

The \textsc{$s$-Hitting Set} problem that asks, given a family of sets $\mathcal{U}$ of size at most $s$ over some universe and a non-negative integer $k$, whether there is a hitting set  $S$ for $\mathcal{U}$ in known to have a polynomial kernel of size at most $(2s-1)k^{s-1}+k$ by the result of Abu-Khzam~\cite{ABUKHZAM2010524}.  Apparently \textsc{Hitting Set} is in {\sf NP}. Hence, there is a polynomial reduction from  \textsc{Hitting Set} to the {\sf NP}-complete problem {\sc Vertex-Removal to $\phi$}. This implies that 
{\sc Vertex-Removal to $\phi$} admits a polynomial kernel. 
\end{proof}

\section{Lower bounds}\label{immature}
Here we prove the hardness claims of Theorems~\ref{magnetic}--\ref{unguided}. In Subsection~\ref{daylight}, we give the technical result about reducing \textsc{Edge-Removal to $\phi$} to \textsc{Vertex-Removal to $\psi$}. In Subsection~\ref{begrudge}, we show ${\sf W}[2]$-hardness and in Subsection~\ref{holiness} we obtain kernelization lower bounds.

\subsection{Reducing Edge Removal to Vertex Removal}\label{daylight}
In this section we construct a generic reduction of \textsc{Edge-Removal to $\phi$} to \textsc{Vertex-Removal to $\psi$} that we use twice in the proofs of our complexity lower bounds. 

We say that an FOL-formula $\phi$
is {\em $\forall$-containing} if the prefix of $\phi$ contains a $\forall$ quantifier.  

\begin{lemma}\label{dispense}
For every $\forall$-containing  FOL-formula $\phi\in \Sigma_{s}\cup \Pi_{s}$ without free variables for $s\geq 1$, there is formula $\psi\in \Sigma_{s+1}\cup \Pi_{s+1}$ without free variables such that
\begin{itemize}
\item[(i)]  $|\psi|={\sf poly}(|\phi|)$,
\item[(ii)]  if $\phi\in\Sigma_s$ (resp. $\phi\in\Pi_s$), then $\phi\in\Sigma_{s+1}$ (resp.  $\phi\in\Pi_{s+1}$),
\item[(iii)] there is a 
polynomial reduction of \textsc{Edge-Removal to $\phi$} to \textsc{Vertex-Removal to $\psi$} that transforms 
each instance $(G,k)$ of  \textsc{Edge-Removal to $\phi$}  to an equivalent instance  $(G',k)$ of \textsc{Vertex-Removal to $\psi$}, i.e.,  the parameter $k$ remains the same.
\end{itemize}
\end{lemma}

\begin{proof} 
 Let $(G,k)$ be an instance of   \textsc{Edge-Removal to $\phi$}.
We construct the instance $(G',k)$ of \textsc{Vertex-Removal to $\psi$} from $(G,k)$
and then we construct $\psi$. The main idea is to replace edge removals by vertex removals switching to the \emph{incidence} graph of $G$ or, equivalently,  by subdividing edges of $G$. Then we have to ``label'' the vertices of the original graph that should not be removed. We do it by making them adjacent to sufficiently many pendant vertices. Formally, we construct $G'$ as follows. 
\begin{itemize}
\item Construct a copy of $G$ and subdivide each edge, that is, for each $e=xy\in E(G)$, delete $e$, construct a new vertex  $v_e$ and make it adjacent to $x$ and $y$. We say that the vertices of $G$ are \emph{branching} vertices and the vertices obtained by the edge subdivisions are called \emph{subdivision} vertices.
\item For each branching vertex $u$, introduce $k+3$ new vertices $v_1,\ldots,v_{k+3}$ and make them adjacent to $u$; we call these vertices \emph{pendant}.
\end{itemize} 
Notice that every subdivision vertex has degree 2 and every branching vertex has degree at least 3. Moreover, if $H$ is obtained from $G'$ by the removal of at most $k$ vertices, then still every remaining branching vertex has degree at least 3. Observe also that $H$ has isolated vertices if and only if at least one branching vertex of $G$ is removed in the construction of $H$.

Our next aim is to construct $\psi$ from $\phi$ to ensure that $(G,k)$ is a {\sf yes}-instance of \textsc{Edge-Removal to $\phi$}  if and only if $(G',k)$ is a {\sf yes}-instance of \textsc{Vertex-Removal to $\psi$}.
Let
$$\phi={\tt Q}_1x_1\ldots{\tt Q}_px_p\chi$$
where ${\tt Q}_1,\ldots,{\tt Q}_p$ are quantifiers, $x_1,\ldots,x_p$ are variables and $\chi$ is quantifier-free. We also assume that $\ \chi$ is written in the conjunctive normal form.

The construction of $\psi$ is done in several steps. First, we take care of adjacencies in $\phi$. Recall that two verices $u$ and $v$ are adjacent in $G$ if and only if they have a common neighbor in $G'$. Respectively, we modify the adjacency predicates in $\chi$.

Let $\Pi=\{\pi_1,\ldots,\pi_s\}$ be the family (multiset) of all predicates of the form $x_i\sim x_j$ that occur in $\ \chi$ without negations. If the same predicate $x_i\sim x_j$ occurs several times, then for each occurrence, we include it in $\Pi$. Similarly, let $\bar{\Pi}=\{\bar{\pi}_1,\ldots,\bar{\pi}_t\}$ be the family (multiset) of all predicates of the form $\neg (x_i\sim x_j)$ that occur in $\ \chi$. Then we do the following.
\begin{itemize}
\item Construct $s$ new variables $y_1,\ldots,y_s$ and $t$ variables $z_1,\ldots,z_t$.
\item For each $h\in \{1,\ldots,s\}$, consider $\pi_h= x_i\sim x_j$ for some $i,j\in\{1,\ldots,p\}$ and replace it by $\neg(x_i=x_j)\wedge (x_i\sim y_h)\wedge (y_h\sim x_j)$.
\item  For each $h\in \{1,\ldots,t\}$, consider $\bar{\pi}_h=\neg (x_i\sim x_j)$ for some $i,j\in\{1,\ldots,p\}$ and replace it by $(x_i=x_j)\vee\neg (x_i\sim y_h)\vee\neg (y_h\sim x_j)$.
\item Denote the formula obtained from $\ \chi$ by $\ \chi'$. Then
\begin{itemize}
\item if ${\tt Q}_p=\exists$, set $\ \sigma=\exists y_1\ldots\exists y_s\forall z_1\dots\forall z_t\ \chi'$, and
\item if ${\tt Q}_p=\forall$, set $\ \sigma=\forall z_1\dots\forall z_t\exists y_1\ldots\exists y_s\ \chi'$.
\end{itemize}
\end{itemize}
Consider the formula
$\alpha={\tt Q}_1x_1\ldots{\tt Q}_px_p\ \sigma$. 
Observe that we added new quantified variables in the end of the prefix of $\phi$ in such a way that we obtain at most one additional alternation of quantifiers. That is, we obtain that $\alpha\in\Sigma_{s+1}$ if $\phi\in\Sigma_s$ and $\alpha\in\Pi_{s+1}$ if $\phi\in\Pi_s$.
The crucial property of the above construction is given in the following straightforward claim. 

\begin{lclaim}\label{connects}
$G\models \phi$ if and only if $G'\models \beta$ where $$\beta={\tt Q}_1x_1\ldots{\tt Q}_px_p\ \sigma  \mbox{ with the domains of the variables $x_{1},\dots,x_{p}$  restricted to } V(G).$$ Moreover, for every set of pendant vertices $X$ of size at most $k$, 
$G\models \phi$ if and only if $(G'-X)\models \beta$. 
\end{lclaim}

Notice that in the formula of Claim~\ref{connects} we insist to  restrict the domains of the variables $x_1,\ldots,x_p$ in $\beta$ to $V(G)$, that is, 
 $\beta$ is not an FOL-formula (and $\beta\neq \alpha$). 
Our next aim is to express these additional constraints in the first-order logic.  We do it using the property that the branching vertices of $G'$ have degrees at least 3 and all the other vertices have degrees at most 2. 

We consecutively construct the formulas $\rho_{n+1},\ldots,\rho_1$. First, we set $\rho_{n+1}=\ \sigma$. Note that $x_1,\ldots,x_p$ are free variable for $\rho_{n+1}$. Assume inductively that $1\leq i\leq p$ and $\rho_{i+1}$ with free variables $x_1,\ldots,x_i$ is already constructed.  Denote by $P_{i+1}$ the prefix and $\mu_{i+1}$ the quantifier-free part  
of $\rho_{i+1}$ respectively, that is, $\rho_{i+1}=P_{i+1}\mu_{i+1}$. The construction of $\rho_i$ depends on the quantifier ${\tt Q}_i$. 
\begin{itemize}
\item Introduce 3 new variables $r_i^1,r_i^2,r_i^3$.
\item If ${\tt Q}_i=\exists$, then set 
\begin{align*}
\rho_i=\exists x_i\exists r_i^1\exists r_i^2\exists r_i^3 P_{i+1}[&(\neg(r_i^1=r_i^2)\wedge \neg(r_i^1=r_i^3)\wedge\neg(r_i^2=r_i^3))\wedge \\
& ((x_i\sim r_i^1) \wedge (x_i\sim r_i^2)\wedge (x_i\sim r_i^3))\wedge \mu_{i+1}].
\end{align*}
\item If ${\tt Q}_i=\forall$, then set 
\begin{align*}
\rho_i=\forall x_i\forall r_i^1\forall r_i^2\forall r_i^3 P_{i+1}[&((\neg(r_i^1=r_i^2)\wedge \neg(r_i^1=r_i^3)\wedge\neg(r_i^2=r_i^3))\wedge \\
& \big(((x_i\sim r_i^1)\wedge (x_i\sim r_i^2)\wedge (x_i\sim r_i^3))\rightarrow \mu_{i+1}\big)].
\end{align*}
\end{itemize}
Let $\gamma=\rho_1$. 
Note that in our construction of $\gamma$, we do not create new alternations of quantifications, that is, $\gamma\in \Sigma_{s+1}$ or in $\Pi_{s+1}$ depending on whether $\alpha\in\Sigma_{s+1}$ or in $\Pi_{s+1}$. 
The construction of $\gamma$ implies the next claim.

\begin{lclaim}\label{handling}
$G\models \phi$ if and only if $G'\models \gamma$. Moreover, for every set of pendant vertices $X$ of size at most $k$, 
$G\models \phi$ if and only if $(G'-X)\models \gamma$. 
\end{lclaim}

Recall that the removal of the edges in $G$ corresponds to the removal of subdivision vertices of $G'$. Respectively, our next aim is to ensure that the removal of a branching vertex of $G'$ leads to the graph for which our formula is false. We use the property that for every set $S$ of at most $k$ vertices of $G'$, $G'-S$ has an isolated vertex if and only if $S$ contains a branching vertex. We use the property that the condition that a graph has no an isolated can be expressed by the formula $\forall s_1\exists s_2\ (s_1\sim s_2)$ and 
modify $\gamma$ as follows.  Let $P$ be the prefix of $\gamma$ and let $\mu$ be the quantifier-free part. We write $P$ as the concatenation of 3 parts $P_1$, $P_2$ and $P_3$ where $P_1$ and/or $P_3$ may be empty. Recall that the prefix of the original formula $\phi$ contains the $\forall x_i$ for some $i\in\{1,\ldots,p\}$ by the condition of the lemma. Hence, the same holds for $\gamma$. Let $P_1$ be the first part of $P$ until the first occurrence of the quantifier $\forall$. Then $P_2$ is the next part until the first occurrence of the quantifier $\exists$ or until the end of $P$ if such a quantifier does not exist. Respectively, $P_3$ is the remaining part. 
We define 
$$\psi=P_1\forall s_1 P_2 \exists s_2 P_3\ [(s_1\sim s_2)\wedge \mu]$$
using two new variables $s_1$ and $s_2$.

Notice that the insertion of the new quantifications is done in such a way that we do not introduce new alternations of quantifiers unless $P_3$ is empty. But if $P_3$ is empty, then 
$\phi\in \Pi_1$ or $\phi\in \Sigma_2$ and the new alternations were not introduced in the construction of $\alpha$ from $\phi$. We have that either $\phi,\gamma\in \Sigma_s$ and $\psi\in \Sigma_{s+1}$ or $\phi,\gamma\in\Pi_s$ and $\psi\in\Pi_{s+1}$.

We show the following claim.

\begin{lclaim}\label{softened}
The instance $(G,k)$ is a {\sf yes}-instance of  \textsc{Edge-Removal to $\phi$} if and only if 
$(G',k)$ is a {\sf yes}-instance of \textsc{Vertex-Removal to $\psi$}. 
\end{lclaim}

To show the claim, assume first that $(G,k)$ is a {\sf yes}-instance of \textsc{Edge-Removal to $\phi$}. Then there is $F\subseteq E(G)$ of size at most $k$ such that 
$(G-F)\models\phi$. We define $S\subseteq V(G')$ be the the set of the subdivision vertices of $G'$ corresponding to the edges of $F$, that is, $S=\{v_e\mid e\in F\}$.
Clearly, $|S|\leq k$. Notice that our reduction algorithm for graphs produces $G'-S$ from $G-F$. Hence, by Claim~\ref{handling}, $(G'-F)\models \gamma$.
Because $G'-S$ has no isolated vertices, we have that $(G'-S)\models \psi$. It means that $(G',k)$ is a {\sf yes}-instance of \textsc{Vertex-Removal to $\psi$}. 

Suppose now that $(G',k)$ is a {\sf yes}-instance of \textsc{Vertex-Removal to $\psi$}. Then there is $S\subseteq V(G')$ of size at most $k$ such that $(G'-S)\models \psi$. Since 
$\psi=P_1\forall s_1 P_2\exists s_2 P_3\ [(s_1\sim s_2)\wedge \mu]$, we have that $G'-S$ has no isolated vertices, that is, $S$ contains only subdivision vertices of $G'$ or pendants vertices. 
Also this immediately implies that $(G'-S)\models\gamma$.  Let $S'$ be the set of subdivision vertices of $S$ and let $X$ be the set of pendant vertices in $S$.  
We define $F$ to be the set of edges of $G$ corresponding to the subdivision vertices in $S'$, that is, $F=\{e\in E(G)\mid v_e\in S'\}$. 
We have that $|F|\leq k$. Let also $X$ be the set of pendant vertices in $S$. Observe again that  our reduction algorithm for graphs produces $G'-S'$ from $G-F$. Then by Claim~\ref{handling}, we have that $(G-F)\models \phi$. We conclude that $(G,k)$ is a {\sf yes}-instance of \textsc{Edge-Removal to $\phi$}. 

\medskip
To complete the proof of the lemma, observe that the size of $\psi$ is polynomial in the size of $\phi$ by our construction of the formula and this shows (i). 
To show (ii), observe  that $\psi\in\Sigma_{s+1}$ if $\phi\in\Sigma_s$ and $\psi\in\Pi_{s+1}$ if $\phi\in\Pi_s$.
The last claim (iii) of the lemma immediately follows from 
Claim~\ref{softened}.
\end{proof}

\subsection{W[2]-hardness}\label{begrudge}
In this subsection we show that there are formulas $\phi$ in $\Pi_2$ and $\Pi_3$
for which 
 \textsc{Edge-Removal (Editing) to $\phi$} and \textsc{Vertex-Removal to $\psi$}. 
respectively are ${\sf W}[2]$-hard when parameterized by $k$. 
First, we show the claim for  \textsc{Edge-Removal to $\phi$} and  \textsc{Edge-Editing to $\phi$}.

\begin{lemma}\label{bringing}
There is an FOL-formula $\phi\in \Pi_2$ without free variables with $5$ variables such that  \textsc{Edge-Removal to $\phi$} and  \textsc{Edge-Editing to $\phi$} are ${\sf W}[2]$-hard. 
\end{lemma}

\begin{proof}
We define the formula $\phi$ as follows:
\begin{align*}
\phi=\forall x\exists y_1\exists y_2\exists y_3\exists y_4[& (x\sim y_1)\wedge (x\sim y_2)\wedge (x\sim y_3)\wedge (x\sim y_4)\wedge (y_1\sim y_2)\wedge(y_1\sim y_3)\wedge\\
&(y_2\sim y_3)\wedge(y_2\sim y_4)\wedge(y_3\sim y_4)\wedge \neg(y_1=y_4)\wedge\neg(y_1\sim y_4)]. 
\end{align*}
In terms of graphs, $G\models\phi$ means that for every vertex $x$ of $G$, there are vertices $y_1,y_2,y_3,y_4$ such that  these vertices together with $x$ induce the graph $W$ shown in Fig.~\ref{operetta}. We say that $W$ is a \emph{$\phi$-witness subgraph rooted in $x$}. 

\begin{figure}[ht]
\centering\scalebox{0.75}{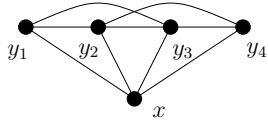}
\caption{The $\phi$-witness graph $W$.
\label{operetta}}
\end{figure} 

We show hardness for \textsc{Edge-Editing to $\phi$} by reducing the \textsc{Set Cover} problem:\medskip

  \defparproblem{\sc Set Cover}{A family of sets $\mathcal{S}$ over the universe $U$ and a positive integer $k$.}{$k$}{Is there a subfamily $\mathcal{S}^*\subseteq \mathcal{S}$ of size at most $k$ that \emph{covers} $U$, that is, every element of $U$ is in one of the set of $\mathcal{S}^*$?
}

\medskip

\noindent It is well-known that \textsc{Set Cover} is ${\sf W}[2]$-hard when parameterized by $k$~\cite{DowneyF13}.

\begin{figure}[ht]
\centering\scalebox{0.75}{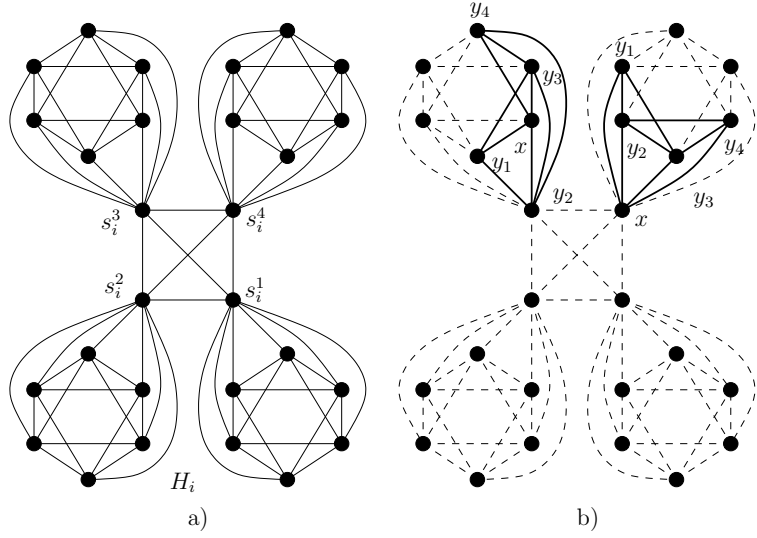}
\caption{Construction of $H_i$ and the witness subgraphs for the vertices of $H_i$. The witness subgraphs are shown by thick lines and the other edges are shown by dashed lines. 
\label{uttering}}
\end{figure} 

Let $(U,\mathcal{S},k)$ be an instance \textsc{Set Cover}, $\mathcal{S}=\{S_1,\ldots,S_m\}$ and $U=\{u_1,\ldots,u_n\}$. We construct the graph $G$ as follows.
\begin{itemize}
\item For each $i\in\{1,\ldots,m\}$, construct the graph $H_i$ with 4 \emph{root vertices} $s_i^1,s_i^2,s_i^3,s_i^4$ as it is shown in Fig.~\ref{uttering} a).
\item For each $j\in\{1,\ldots,n\}$, construct $2k+1$ vertices $u_j^0,\ldots,u_j^{2k}$ and make them adjacent to the root vertices of all the gadgets $H_r$ such that the element $u_j$ of the universe $U$ is in the set $S_r\in \mathcal{S}$. 
\end{itemize}

We claim that $(U,\mathcal{S},k)$ is a {\sf yes}-instance \textsc{Set Cover} if and only if $(G,k)$ is a {\sf yes}-instance of \textsc{Edge-Editing to $\phi$}.  

Suppose that $(U,\mathcal{S},k)$ is a {\sf yes}-instance \textsc{Set Cover}. Let $\mathcal{S}^*\subseteq \mathcal{S}$ be a family of size at most $k$ that covers $U$. We construct the set of edges $F$ of $G$ as follows. For every $S_i\in\mathcal{S}^*$, we include the edge $s_i^1s_i^4$ of the gadget $H_i$ in $F$. Clearly, $|F|\leq k$. Let $G'=G\bigtriangleup F=G-F$. We show that $G'\models \phi$. Recall that we have to show that for every $x\in V(G')$, there are $y_1,y_2,y_3,y_4$, such that $G[\{x,y_1,y_2,y_3,y_3\}]$ is a $\phi$-witness subgraph rooted in $x$. 
For $x\in\bigcup_{i=1}^mV(H_i)$, such subgraphs are shown in Fig.~\ref{uttering} b). Let $x=u_j^h$ for $j\in\{1,\ldots,n\}$ and $h\in\{0,\ldots,2k\}$. The element $u_j\in U$ is covered by some set $S_i\in\mathcal{S}^*$. Since $s_i^1s_i^4\in F$, we have that $G[u_j^h,s_i^1,s_i^2,s_i^3,s_i^4]$ is  a $\phi$-witness subgraph rooted in $x$. We conclude that $G'\models \phi$ and, therefore, 
$(G,k)$ is a {\sf yes}-instance of \textsc{Edge-Editing to $\phi$}.

Assume that $(G,k)$ is a {\sf yes}-instance of \textsc{Edge-Editing to $\phi$}. Then there is $F\subseteq \binom{V(G)}{2}$ with $|F|\leq k$ such that for $G'=G\bigtriangleup F$, it holds that $G'\models\phi$. Consider the auxiliary graph $Q=(V(G),F)$. For $i\in\{1,\ldots,m\}$, let $\delta_i=\sum_{v\in V(H_i)}d_Q(v)$. We define $\mathcal{S}^*=\{S_i\mid 1\leq i\leq m,~\delta_i\geq 2\}$. Because $|F|\leq k$, $\sum_{i=1}^n\delta_i\leq 2k$ and, therefore, $|\mathcal{S}^*|\leq k$. We claim that $\mathcal{S}^*$ covers $U$.
To obtain a contradiction, assume that there is $j\in\{1,\ldots,n\}$ such that $u_j$ is not covered by $\mathcal{S}^*$.  Since $|F|\leq k$, there is $h\in\{0,\ldots,2k\}$ such that the vertex $u_j^h$ is not incident to the pairs of $F$. Because $G'\models \phi$, there is a  $\phi$-witness subgraph rooted in $x=u_j^h$. Hence, there are $y_1,y_2,y_3,y_4\in V(G)$ such that $G[\{x,y_1,y_2,y_3,y_4\}]$ is a  $\phi$-witness subgraph. Notice that $y_1,y_2,y_3,y_4\in \cup_{i=1}^mV(H_i)$. Observe that if $y_s\in V(H_i)$, then $F\cap\binom{V(H_i)}{2}=\emptyset$, because $\delta_i\leq 1$. Hence, it cannot happen that $y_1,y_2,y_3,y_4\in V(H_i)$ for some $i\in\{1,\ldots,m\}$, because it would mean that $\{y_1,y_2,y_3,y_4\}=\{s_i^1,s_i^2,s_i^3,s_i^4\}$ but $G'[\{s_i^1,s_i^2,s_i^3,s_i^4\}]=K_4$, a contradiction. Therefore, there are distinct $i,i'\in\{1,\ldots,n\}$ such that $V(H_i)\cap\{y_1,y_2,y_3,y_4\}\neq\emptyset$ and $V(H_{i'})\cap\{y_1,y_2,y_3,y_4\}\neq\emptyset$.   Since $\delta_i,\delta_{i'}\leq 1$, there is a unique pair $ab$ of $F$ such that $a\in V(H_i)$ and $b\in V(H_{i'})$. Moreover, $ab$ is a bridge of $G'[\cup_{s=1}^mV(H_s)]$ and, therefore, $ab$ is a bridge of $G'[\{y_1,y_2,y_3,y_4\}]$. This contradicts the fact that  $W-x$ is 2-connected. We conclude that    
$\mathcal{S}^*$ covers $U$. Hence,  $(U,\mathcal{S},k)$ is a {\sf yes}-instance \textsc{Set Cover}.

This concludes  the ${\sf W}[2]$-hardness proof for \textsc{Edge-Editing to $\phi$}. To show that  \textsc{Edge-Removal to $\phi$} is ${\sf W}[2]$-hard when parameterized by $k$, we use the same reduction. Note that to show that if $(U,\mathcal{S},k)$ is a {\sf yes}-instance of \textsc{Set Cover}  then $(G,k)$ is a {\sf yes}-instance of \textsc{Edge-Editing to $\phi$}, we constructed $F\subseteq E(G)$, that is, we proved that $(G,k)$ is a {\sf yes}-instance of \textsc{Edge-Removal to $\phi$}.
\end{proof}

Lemma~\ref{bringing} and Observation~\ref{adorning} imply Theorem~\ref{narcotic} (ii). To show the claim for \textsc{Vertex-Removal to $\phi$} we 
combain Lemma~\ref{bringing} with Lemma~\ref{dispense} and obtain the following lemma that implies Theorem~\ref{magnetic} (ii).

\begin{lemma}\label{furrowed} There is a constant $c$ such that 
there is an FOL-formula $\phi\in\Pi_3$ without free variables that has  at most $c$ variables  such that  \textsc{Vertex-Removal to $\phi$} is ${\sf W}[2]$-hard. 
\end{lemma}

\subsection{Kernelization lower bounds}\label{holiness}
In this subsection we obtain the kernelization lower bounds for   \textsc{Edge-Removal (Editing) to $\phi$} and \textsc{Vertex-Removal to $\psi$}.

First, we show the lower bounds for  \textsc{Edge-Removal to $\phi$} and  \textsc{Edge-Editing to $\phi$}  for $\Pi_1$-formulas.
To do it, we use the known results about kernelization lower bounds for the \textsc{$H$-Free Edge Removal} and \textsc{$H$-Free Editing}. Recall that  for a graph $H$, \textsc{$H$-Free Edge Removal} (\textsc{$H$-Free Editing}) asks, given a graph $G$ and a nonnegative integer $k$, whether there is a set of edges $F$ (a set $F\subseteq\binom{V(G)}{2}$ respectively) of size at most $k$ such that $G-F$ ($G\bigtriangleup F$ respectively) does not contain an induced subgraph isomorphic to $H$. Since the property that a graph $G$ has no induced subgraph isomorphic to $H$ can be expressed by an FOL-formula $\phi_H\in \Pi_1$  that has $|V(H)|$ variables,  
 \textsc{$H$-Free Edge Removal} and \textsc{$H$-Free Editing} can be written as  \textsc{Edge-Removal to $\phi_H$} and  \textsc{Edge-Editing to $\phi_H$}
 respectively.
 The first kernelization lower bounds for   \textsc{$H$-Free Edge Removal} and \textsc{$H$-Free Editing} were obtained by Kratsch and Wahlstr{\"{o}}m in~\cite{KratschW13} who proved that there are graphs $H$ for which these problems do not admit polynomial kernels unless ${\sf NP}\subseteq{\sf coNP}/{\sf poly}$. Some further results were obtained by Guillemot et al.~\cite{GuillemotHPP13}. In~\cite{CaiC15} Cai and Cai  completely characterized the cases when the problems have no polynomial kernels if $H$ is a path or cycle or is 3-connected graph up to the conjecture that ${\sf NP}\not\subseteq{\sf coNP}/{\sf poly}$.  In particular, they proved that \textsc{$H$-Free Edge Removal} and \textsc{$H$-Free Editing} do not have polynomial kernels if $H=C_4$ unless ${\sf NP}\subseteq{\sf coNP}/{\sf poly}$. This immediately yields the following lemma.

\begin{lemma}
 \label{supports}
There is an FOL-formula $\phi\in\Sigma_1$ without free variables that has $5$ variables such that  
\textsc{Edge-Removal to $\phi$} and  \textsc{Edge-Editing to $\phi$}
have no polynomial kernels  unless  ${\sf NP}\subseteq{\sf coNP}/{\sf poly}$.
\end{lemma}

Lemma~\ref{supports} and Observation~\ref{adorning} prove Theorem~\ref{unguided} (ii). 
Using Lemma~\ref{dispense}, we obtain the following lemma for \textsc{Vertex-Removal to $\phi$}.

\begin{lemma}\label{conceive} 
There is a constant $c$ such that 
there is an FOL-formula $\phi\in\Pi_2$ without free variables that has  at most $c$ variables  such that  \textsc{Vertex-Removal to $\phi$} is  has no polynomial kernel unless  ${\sf NP}\subseteq{\sf coNP}/{\sf poly}$.
\end{lemma}
 
Our final task is to show that it is unlikely that  \textsc{Vertex-Removal to $\phi$} has a polynomial kernel for $\Sigma_2$-formulas. We do it by using the {\sl cross-composition} technique introduced by Bodlaender, Jansen and Kratsch~\cite{BodlaenderJK14} (see also~\cite{CyganFKLMPPS15} for the introduction to the technique). Here we only briefly sketch the main notions that we need to apply it.

Let $\Sigma$ be a finite alphabet.  An equivalence relation $\mathcal{R}$ on the set of strings $\Sigma^*$ is called a \emph{polynomial equivalence relation} if the following two conditions hold:
\begin{itemize}
\item[i)] there is an algorithm that given two strings $x,y\in\Sigma^*$ decides whether $x$ and $y$ belong to
the same equivalence class in time polynomial in $|x|+|y|$,
\item[ii)] for any finite set $S\subseteq\Sigma^*$, the equivalence relation $\mathcal{R}$ partitions the elements of $S$ into a
number of classes that is polynomially bounded in the size of the largest element of $S$.
\end{itemize}

Let $\mathcal{L}\subseteq\Sigma^*$ be a problem, let $\mathcal{R}$ be a polynomial
equivalence relation on $\Sigma^*$, and let $\mathcal{P}\subseteq\Sigma^*\times\mathbb{N}$   
be a parameterized problem.  An \emph{OR-cross-composition of $\mathcal{L}$ into $\mathcal{P}$} (with respect to $\mathcal{R}$) is an algorithm that, given $t$ instances $I_1,I_2,\ldots,I_t\in\Sigma^*$ 
of $\mathcal{L}$ belonging to the same equivalence class of $\mathcal{R}$, takes time polynomial in
$\sum_{i=1}^t|I_i|$ and outputs an instance $(I,k)\in \Sigma^*\times \mathbb{N}$ such that:
\begin{itemize}
\item[i)] the parameter value $k$ is polynomially bounded in $\max\{|I_1|,\ldots,|I_t|\} + \log t$,
\item[ii)] the instance $(I,k)$ is a {\sf yes}-instance of $\mathcal{P}$ if and only there is $i\in\{1,\ldots,t\}$ such that $I_i$ is a {\sf yes}-instance of $\mathcal{L}$.
\end{itemize}
It is said that $\mathcal{L}$ \emph{OR-cross-composes into} $\mathcal{P}$ if a cross-composition
algorithm exists for a suitable relation $\mathcal{R}$.

Bodlaender, Jansen and Kratsch~\cite{BodlaenderJK14} proved the following theorem.

\begin{theorem}[\!\!\!\cite{BodlaenderJK14}]\label{typifies}
If an {\sf NP}-hard problem $\mathcal{L}$ OR-cross-composes into the parameterized problem $\mathcal{P}$,
then $\mathcal{P}$ does not admit a polynomial kernelization unless
${\sf NP}\subseteq{\sf coNP}/{\sf poly}$.
\end{theorem}

 \begin{lemma}\label{solution} There is  $\phi\in \Sigma_2$ without free variables that has  $3$ variables such that  \textsc{Vertex-Removal to $\phi$} has no polynomial kernel  unless  ${\sf NP}\subseteq{\sf coNP}/{\sf poly}$.
\end{lemma}

\begin{proof}
We define the formula $\phi$ as follows:
\begin{equation*}
\phi=\exists x\forall y\forall z[((x\sim y)\wedge(x\sim z))\rightarrow((y=z)\vee(y\sim z))].
\end{equation*}
In terms of graphs, $G\models\phi$ means that there is a vertex $x$ whose neighborhood is a clique.

We consider the  \textsc{Clique} problem:\medskip

  \defproblemu{{\sc Clique}}{A graph $G$ and a positive integer $k$.}{Is there a clique in $G$ with at least $k$ vertices?}

\medskip

\noindent and show that \textsc{Clique} OR-cross-composes into \textsc{Vertex-Removal to $\phi$}.

We say that two instances $(G_1,k_1)$ and $(G_2,k_2)$ of \textsc{Clique} are \emph{equivalent} if $|V(G_1)|=|V(G_2)|$ and $k_1=k_2$. 

Let $(G_1,k),\ldots,(G_t,k)$ be equivalent instances of \textsc{Clique} where graphs have $n$ vertices. We construct the instance $(G',k')$ of  \textsc{Vertex-Removal to $\phi$} as follows.
\begin{itemize}
\item Construct disjoint copies of $G_1,\ldots,G_t$.
\item For every $i\in\{1,\ldots,t\}$, construct $n-k+2$ vertices $u_i^1,\ldots,u_i^{n-k+2}$ and make them adjacent to the vertices of $G_i$.
\item Set $k'=n-k$.
\end{itemize}

We claim that $(G',k')$ is a {\sf yes}-instance of \textsc{Vertex-Removal to $\phi$} if and only if there is $i\in\{1,\ldots,t\}$ such that $(G_i,k)$ is a {\sf yes}-instance of \textsc{Clique}.

Suppose that there is $i\in\{1,\ldots,t\}$ such that $(G_i,k)$ is a {\sf yes}-instance of \textsc{Clique}. Then $G_i$ has a clique $K$ of size $k$. Let $S=V(G)\setminus K$. Note that $|S|=n-k=k'$. Now for $x=u_i^1$, we have that the neighborhood of $x$ in $G'$ is the clique $K$, that is, $(G',k')$ is a {\sf yes}-instance of \textsc{Vertex-Removal to $\phi$}.

Assume that $(G',k')$ is a {\sf yes}-instance of \textsc{Vertex-Removal to $\phi$}. Then there is a set of vertices $S\subseteq V(G')$ of size at most $k'$ such that $(G'-S)\models\phi$. 
Let $G''=G'-S$. We have that there is $x\in V(G'')$ such that the neighborhood of $x$ in $G''$ is a clique. Then there is $i\in\{1,\ldots,t\}$ such that $x\in V(G_i)$ or $x\in\{u_i^1,\ldots,u_i^{n-k+2}\}$. Suppose that $x\in V(G_i)$. Since $|S|\leq k'=n-k$, $x$ is adjacent in $G''$ to at least two distinct vertices of $\{u_i^1,\ldots,u_i^{n-k+2}\}$ but these two verices are not adjacent. It implies that $x\in\{u_i^1,\ldots,u_i^{n-k+2}\}$ and the neighborhood of $x$ in $G''$ is $V(G_i)\setminus S$, that is, $K=V(G_i)-S$ is a clique. Because $|S|\leq k'$, we have that $|K|\geq n-k'=k$, that is, $(G_i,k)$ is a {\sf yes}-instance of \textsc{Clique}.

Since $|V(G')|=\Oh(nt)$ and $k'=\Oh(n)$, we conclude that \textsc{Vertex-Removal to $\phi$} has no polynomial kernel 
unless  ${\sf NP}\subseteq{\sf coNP}/{\sf poly}$ by Theorem~\ref{typifies}.
\end{proof}
 
We have that Lemmata~\ref{conceive} and \ref{solution} imply Theorem~\ref{consiste} (ii).

\section{Conclusion}\label{barracks}
In this paper we have provided necessary and sufficient conditions (subject to some  complexity assumptions) on
the fixed-parameter tractability, as well as polynomial kernelization, of graph modification problems to the properties expressible by an FOL-formula from a certain 
 prefix class.
  While we stated our results for undirected graphs, in fact, all our results could be rewritten for directed graphs.  In particular, the {\sf FPT} and kernelization algorithms work for directed graphs without any changes. For the hardness proofs, we need only  a minor modification.
Denote by $\Arc(x,y)$ the predicate for variables $x$ and $y$ meaning that $(x,y)$ is an arc of a directed graph.
Denote by $\vec{\phi}$ the FOL-formula on directed graphs obtained from an FOL-formula $\phi$ on undirected graphs by replacing every predicate $x \sim y$  with $\Arc(x,y)\vee\Arc(y,x)$. 
Then we have the following observation.

\begin{observation}\label{straight}
Let $G$ be the underlaying undirected graph of a directed graph $D$ and let $\phi$ be an FOL-formula on undirected graphs without free variables. Then $G\models\phi$ if and only if $D\models \vec{\phi}$.
\end{observation}

Observation~\ref{straight} immediately implies that whenever  \textsc{Vertex Removal to $\phi$} or \textsc{Edge Removal/Completion/Editing to $\phi$} is hard ({\sf W[2]}-hard or does not have a polynomial kernel unless ${\sf NP}\subseteq{\sf coNP}/{\sf poly}$), the same holds for the variant of the problem on  directed graphs.
The straightforward reduction constructs a directed graph from an undirected graph $G$ by turning its edges to arcs by assigning arbitrary orientations. 

 Our results are for FOL-formulas. It would be very interesting to obtain a similar type of dichotomies for prefix classes of  \emph{Monadic Second Order Logic} (MSOL) formulas on graphs. MSOL is substantially richer and allows to express more interesting graph properties like connectivity that cannot be expressed in FOL. The crucial difference is that while \textsc{Model Checking} for FOL-formulas can be solved in polynomial time for formulas of bounded size  (see Theorem~\ref{appetite}), the problem for MSOL is well-known to be {\sf NP}-complete even for formulas whose size is bounded by a constant.

\paragraph*{Acknowledgments.} We are grateful to  P{\aa}l Drange for his very helpful remarks.


\begin{thebibliography}{10}

\bibitem{ABUKHZAM2010524}
{\sc F.~N. Abu-Khzam}, {\em A kernelization algorithm for d-hitting set},
  Journal of Computer and System Sciences, 76 (2010), pp.~524 -- 531.

\bibitem{BodlaenderDFH09}
{\sc H.~L. Bodlaender, R.~G. Downey, M.~R. Fellows, and D.~Hermelin}, {\em On
  problems without polynomial kernels}, J. Comput. Syst. Sci., 75 (2009),
  pp.~423--434.

\bibitem{BodlaenderJK14}
{\sc H.~L. Bodlaender, B.~M.~P. Jansen, and S.~Kratsch}, {\em Kernelization
  lower bounds by cross-composition}, {SIAM} J. Discrete Math., 28 (2014),
  pp.~277--305.

\bibitem{borger2001classical}
{\sc E.~B{\"o}rger, E.~Gr{\"a}del, and Y.~Gurevich}, {\em The classical
  decision problem}, Springer Science \& Business Media, 2001.

\bibitem{CaiC15}
{\sc L.~Cai and Y.~Cai}, {\em Incompressibility of {$H$}-free edge modification
  problems}, Algorithmica, 71 (2015), pp.~731--757.

\bibitem{CyganFKLMPPS15}
{\sc M.~Cygan, F.~V. Fomin, L.~Kowalik, D.~Lokshtanov, D.~Marx, M.~Pilipczuk,
  M.~Pilipczuk, and S.~Saurabh}, {\em Parameterized Algorithms}, Springer,
  2015.

\bibitem{DowneyF13}
{\sc R.~G. Downey and M.~R. Fellows}, {\em Fundamentals of Parameterized
  Complexity}, Texts in Computer Science, Springer, 2013.

\bibitem{fagin7generalized}
{\sc R.~Fagin}, {\em Generalized first-order spectra and polynomial-time
  recognizable sets}, in Complexity of Computation, vol.~7, AMS, 1974,
  pp.~43--74.

\bibitem{FlumG06}
{\sc J.~Flum and M.~Grohe}, {\em Parameterized Complexity Theory}, Texts in
  Theoretical Computer Science. An {EATCS} Series, Springer, 2006.

\bibitem{FrickG04}
{\sc M.~Frick and M.~Grohe}, {\em The complexity of first-order and monadic
  second-order logic revisited}, Ann. Pure Appl. Logic, 130 (2004), pp.~3--31.

\bibitem{DBLP:journals/jacm/GottlobKS04}
{\sc G.~Gottlob, P.~G. Kolaitis, and T.~Schwentick}, {\em Existential
  second-order logic over graphs: Charting the tractability frontier}, J.
  {ACM}, 51 (2004), pp.~312--362.

\bibitem{grohe2017descriptive}
{\sc M.~Grohe}, {\em Descriptive complexity, canonisation, and definable graph
  structure theory}, vol.~47, Cambridge University Press, 2017.

\bibitem{GuillemotHPP13}
{\sc S.~Guillemot, F.~Havet, C.~Paul, and A.~Perez}, {\em On the
  (non-)existence of polynomial kernels for {$P_\ell$}-free edge modification
  problems}, Algorithmica, 65 (2013), pp.~900--926.

\bibitem{ImpagliazzoPZ01}
{\sc R.~Impagliazzo, R.~Paturi, and F.~Zane}, {\em Which problems have strongly
  exponential complexity?}, J. Comput. Syst. Sci., 63 (2001), pp.~512--530.

\bibitem{KratschW13}
{\sc S.~Kratsch and M.~Wahlstr{\"{o}}m}, {\em Two edge modification problems
  without polynomial kernels}, Discrete Optimization, 10 (2013), pp.~193--199.

\bibitem{lewis1980nodedeletion}
{\sc J.~M. Lewis and M.~Yannakakis}, {\em The node-deletion problem for
  hereditary properties is {NP}-complete}, J. Comput. Syst. Sci., 20 (1980),
  pp.~219--230.

\bibitem{Niedermeierbook06}
{\sc R.~Niedermeier}, {\em Invitation to fixed-parameter algorithms}, vol.~31
  of Oxford Lecture Series in Mathematics and its Applications, Oxford
  University Press, 2006.

\bibitem{Smorynski77}
{\sc C.~Smorynski}, {\em The incompleteness theorems}, in Handbook of
  mathematical logic, vol.~90 of Stud. Logic Found. Math., North-Holland,
  Amsterdam, 1977, pp.~821--865.

\bibitem{Vardi82}
{\sc M.~Y. Vardi}, {\em The complexity of relational query languages (extended
  abstract)}, in Proceedings of the 14th Annual {ACM} Symposium on Theory of
  Computing, May 5-7, 1982, San Francisco, California, {USA}, {ACM}, 1982,
  pp.~137--146.

\bibitem{Williams14}
{\sc R.~Williams}, {\em Faster decision of first-order graph properties}, in
  Joint Meeting of the Twenty-Third {EACSL} Annual Conference on Computer
  Science Logic {(CSL)} and the Twenty-Ninth Annual {ACM/IEEE} Symposium on
  Logic in Computer Science (LICS), {CSL-LICS} '14, Vienna, Austria, July 14 -
  18, 2014, {ACM}, 2014, pp.~80:1--80:6.

\bibitem{yannakakis1981edgedeletion}
{\sc M.~Yannakakis}, {\em Edge-deletion problems}, {SIAM} Journal on Computing,
  10 (1981), pp.~297--309.

\end{thebibliography}

\end{document}